\documentclass[onefignum,onetabnum]{siamart171218} % siam
\usepackage{amsfonts}
\usepackage{graphicx}
\usepackage{epstopdf}
\usepackage{algorithmic}
\usepackage{amssymb}
\usepackage{mathtools}
\usepackage{color}
\usepackage[shortlabels]{enumitem}
\usepackage{multirow}
\usepackage{xfrac}
\usepackage{hyperref}
\usepackage{arydshln}
\usepackage[super]{nth}
\usepackage{cite}
\usepackage{icomma}

\allowdisplaybreaks

\DeclarePairedDelimiter\floor{\lfloor}{\rfloor}

\ifpdf
  \DeclareGraphicsExtensions{.eps,.pdf,.png,.jpg}
\else
  \DeclareGraphicsExtensions{.eps}
\fi

\renewcommand{\proofbox}{$\natural$}

\newsiamremark{remark}{Remark}
\newsiamthm{conjecture}{Conjecture}

\headers{Bohemian Upper Hessenberg Toeplitz Matrices}{E. Y. S. Chan, et al.}

\title{Bohemian Upper Hessenberg Toeplitz Matrices}

\author{
  Eunice Y. S. Chan\thanks{Department of Applied Mathematics, Western University
    (\email{echan295@uwo.ca}, 
    \email{rcorless@uwo.ca},
    \email{sthornt7@uwo.ca}).}
  \and
  Robert M. Corless\footnotemark[1]
  \and
  Laureano Gonzalez-Vega\thanks{Departamento de Matematicas, Estadistica y Computacion, Universidad de Cantabria
  (\email{laureano.gonzalez@unican.es}).}
  \and
  J.~Rafael Sendra\thanks{Research Group ASYNACS, Departamento de F{\'i}sica y Matem\'{a}ticas, University of Alcal{\'a}
  (\email{rafael.sendra@uah.es}).}
  \and
  Juana Sendra\thanks{Universidad Polit{\'{e}}cnica de Madrid
  (\email{jsendra@etsist.upm.es}).}
  \and
  Steven E. Thornton\footnotemark[1]
}

\usepackage{amsopn}

\begin{document}

\maketitle

% ============================================================================ %
% Abstract                                                                     %
% ============================================================================ %
\begin{abstract}
We look at Bohemian matrices, specifically those with entries from $\{-1, 0, {+1}\}$. More, we specialize the matrices to be upper Hessenberg, with subdiagonal entries $1$. Even more, we consider Toeplitz matrices of this kind. Many properties remain after these specializations, some of which surprised us.
Focusing on only those matrices whose characteristic polynomials have maximal height allows us to explicitly identify these polynomials and give a lower bound on their height. This bound is exponential in the order of the matrix.
\end{abstract}

% ============================================================================ %
% Introduction                                                                 %
% ============================================================================ %
\section{Introduction}
A matrix family is called \textbf{Bohemian} if its entries come from a fixed finite discrete (and hence bounded) set, usually integers. The name is a mnemonic for \textbf{Bo}unded \textbf{He}ight \textbf{M}atrix of \textbf{I}ntegers. Such families arise in many applications (e.g. compressed sensing) and the properties of matrices selected ``at random'' from such families are of practical and mathematical interest. An overview of some of our original interest in Bohemian matrices can be found in~\cite{corless2017bohemian}.

We began our study by considering Bohemian upper Hessenberg matrices. We proved two recursive formulae for the characteristic polynomials of upper Hessenberg matrices (see~\cite{chan2018BUH} for details). During the course of our computations, we encountered ``maximal polynomial height'' characteristic polynomials when the matrices were not only upper Hessenberg, but Toeplitz ($h_{i,j}$ constant along diagonals $j-i = k$). Further restrictions to this class allowed identification of key results including explicit formulae for the characteristic polynomials of maximal height, which motivates this paper. In what follows, we lay out definitions and prove several facts of interest about characteristic polynomials and their respective height for these families.

In Figure~\ref{fig:UHT_14}, we see all the eigenvalues of all $14 \times 14$ upper Hessenberg Toeplitz matrices with subdiagonal entries equal to $1$ and all other entries from the population $\{-1, 0, {+1}\}$. We see a wide irregularly hexagonal shape. In contrast, upper Hessenberg Bohemian matrices that are \textsl{not} Toeplitz generate an irregular octagonal shape (see~\cite{chan2018BUH}).  More, the density of eigenvalues (here, a darker colour indicates higher density of eigenvalues) is quite irregular, with high-density flecks dispersed throughout. In some ways the picture is reminiscent of seeds in a cotton ball, if the cotton ball has been flattened. The conjugate symmetry and $z \to -z$ symmetry are evident; to save space, we could have plotted only the first quadrant, but for completeness have included all four. This helps to show that there is a slightly lower density of eigenvalues near (not on) the real line.  The density of eigenvalues actually \textsl{on} the real line is quite high, although this is not evident from the picture.

The one thing that is easily explained about that figure is the wide flat top (and bottom). To do this, consider eigenvalues of Bohemian Upper Hessenberg Toeplitz matrices with \textsl{zero diagonal}.  
Figure~\ref{fig:UHT_14_0_Diag} is a picture of the set of eigenvalues of all $14 \times 14$ upper Hessenberg Toeplitz matrices with subdiagonal entries equal to 1, diagonal entries equal to 0, and all other entries from the population $\{-1, 0, {+1}\}$. Here, we also see a hexagonal shape, but this time, it is not as wide. The matrices $B$ giving rise to Figure~\ref{fig:UHT_14} are exactly the matrices $B=A$, $B = A + I$ and $B = A-I$
where the matrices $A$ give rise to Figure~\ref{fig:UHT_14_0_Diag}; thus the eigenvalues of each $A$ occur three times, once with zero shift, once with $-1$ shift, and once with $1$ shift.  That is, Figure~\ref{fig:UHT_14} is simply three copies of Figure~\ref{fig:UHT_14_0_Diag} placed side by side, giving the appearance of a flat (or mostly flat) top and bottom.

In Figure~\ref{fig:UHT_14_0_Diag} we see more clearly that the high-density ``flecks'' occur moderately near to the edge of the eigenvalue inclusion region.  We have no explanation for this. We also see that the eigenvalues fit into a rough diamond shape; one wonders if the eigenvalues $\lambda = x + iy$ fit into a region of shape $|x|+|y| \le O(\sqrt{n})$.  Again, we have no explanation for this (or even much data; we do not know if this guess is even correct experimentally).

In this paper we seek to explain some other features of these pictures, and to learn more about Bohemian upper Hessenberg Toeplitz matrices. We provide supplementary material through a git repository available at \url{https://github.com/BohemianMatrices/Bohemian_Upper_Hessenberg_Toeplitz_Matrices}. This repository provides all code and data used to generate the results, figures, and tables in this paper.

\begin{figure}
    \centering
    \includegraphics[width=\textwidth]{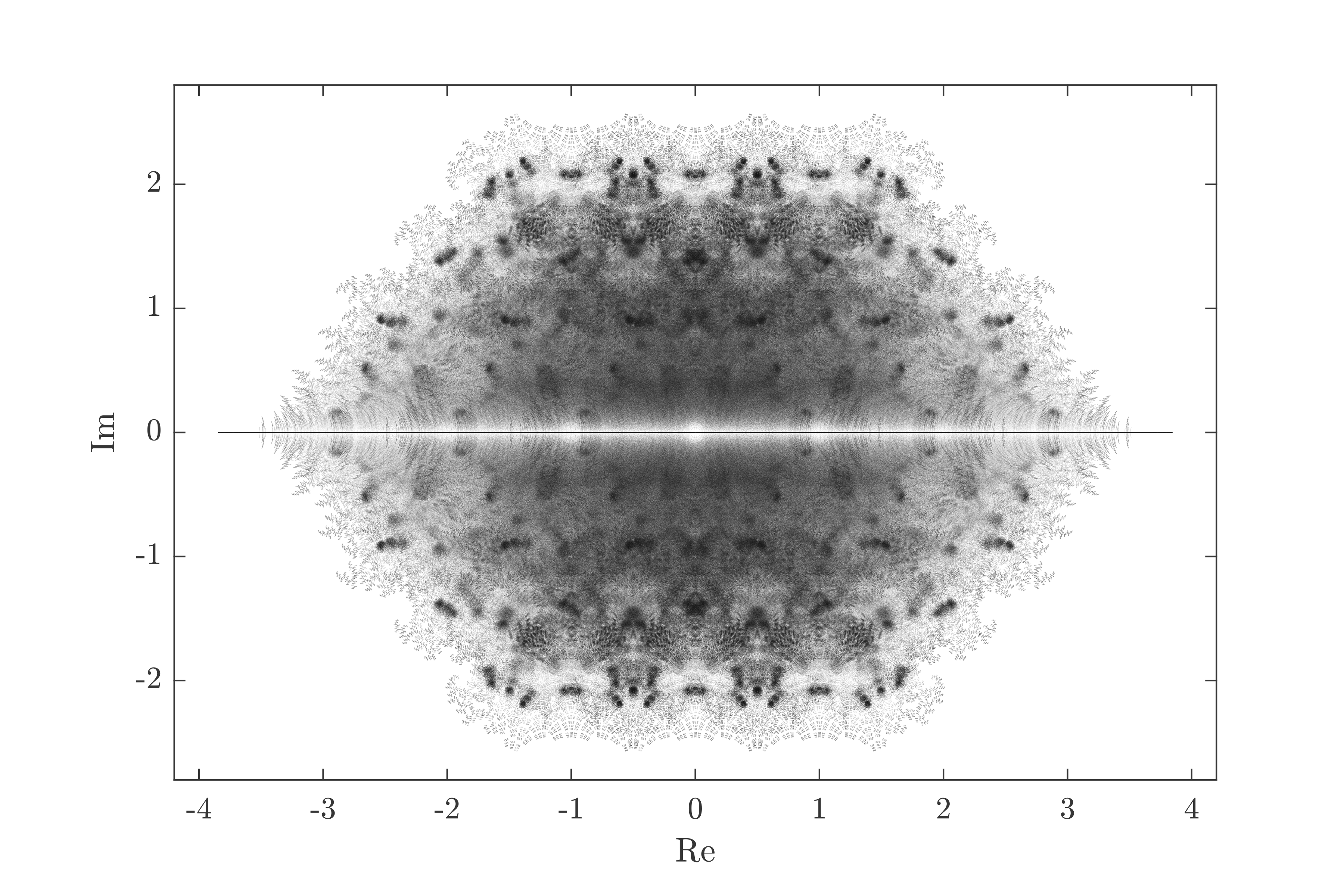}
    \caption{The set of eigenvalues of all $14 \times 14$ upper Hessenberg Toeplitz matrices with subdiagonal entries equal to $1$, and all other entries from the set $\{-1, 0, {+1}\}$. A more detailed image can be found at \url{assets.bohemianmatrices.com/gallery/UHT_14x14.png}}
\label{fig:UHT_14}
\end{figure}

\begin{figure}
    \centering
    \includegraphics[width=\textwidth]{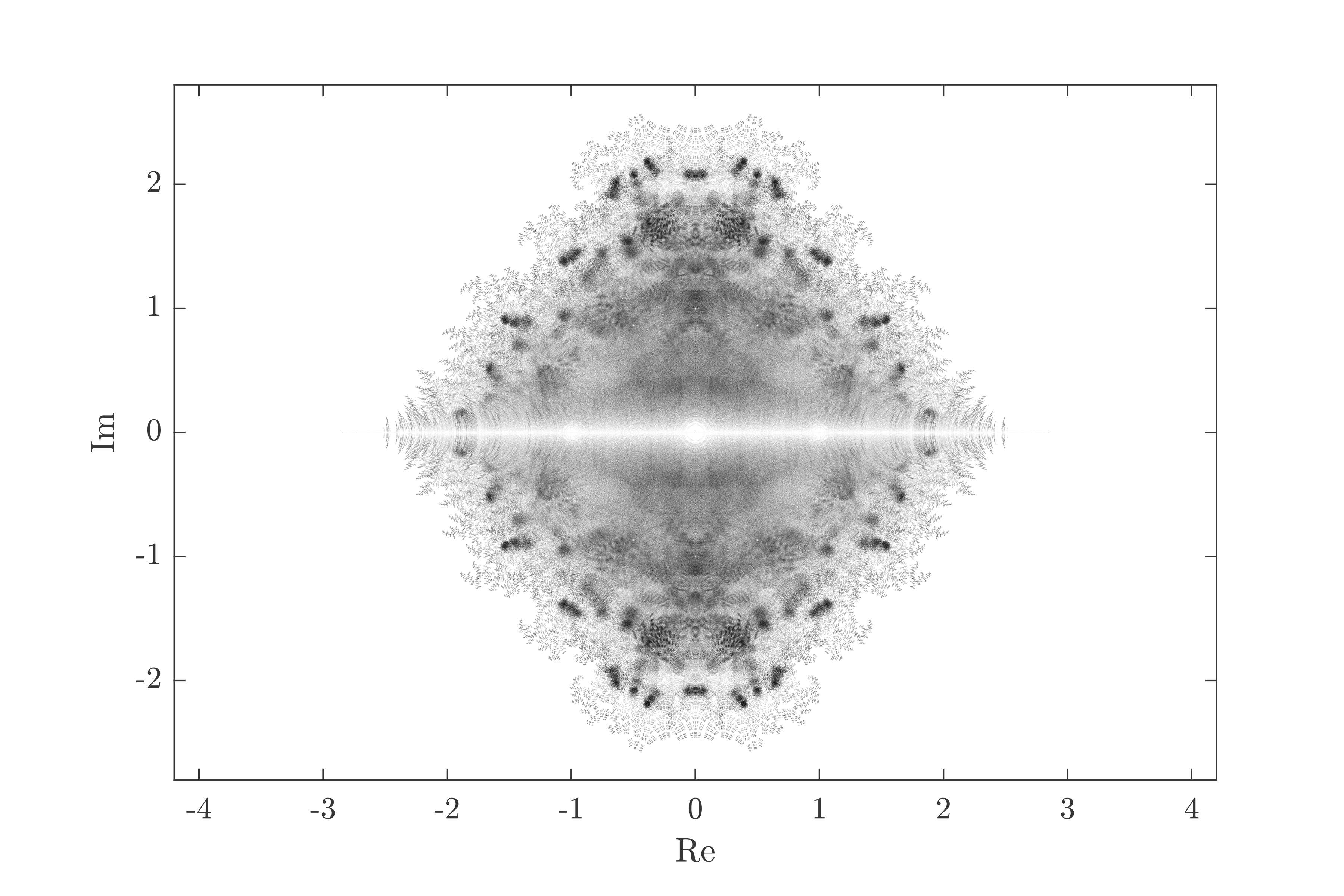}
    \caption{The set of eigenvalues of all $14 \times 14$ upper Hessenberg Toeplitz 
    matrices subdiagonal entries equal to $1$, diagonal entries equal to $0$, and all other entries from the set $\{-1, 0, {+1}\}$. A more detailed image can be found at \url{assets.bohemianmatrices.com/gallery/UHT_0_Diag_14x14.png}}
\label{fig:UHT_14_0_Diag}
\end{figure}

% ============================================================================ %
% Prior Work                                                                   %
% ============================================================================ %
\section{Prior Work}

In our sister paper ``Bohemian Upper Hessenberg Matrices"~\cite{chan2018BUH}, we introduced the following theorems, definitions, remarks, and propositions for upper Hessenberg Bohemian matrices of the form
\begin{equation}
    \mathbf{H}_n =
    \renewcommand{\arraystretch}{1.3}
    \begin{bmatrix}
        h_{1,1} & h_{1,2}    & h_{1,3}     & \cdots & h_{1,n}\\
        s   & h_{2,2}    & h_{2,3}    & \cdots & h_{2,n}\\
        0     & s    & h_{3,3}    & \cdots & h_{3,n}\\
        \vdots & \ddots & \ddots & \ddots & \vdots\\
        0      & \cdots & 0       & s    & h_{n,n}
    \end{bmatrix}
\end{equation}
with characteristic polynomial $Q_n(z) \equiv \det(z\mathbf{I} - \mathbf{H}_n)$.

\begin{definition}
    The set of all $n \times n$ Bohemian upper Hessenberg matrices with upper triangle population $P$ and subdiagonal population from a discrete set of roots of unity, say $s\in \{e^{i\theta_{k}}\}$ where $\{\theta_{k}\}$ is some finite set of angles, is called $\mathcal{H}_{\{\theta_{k}\}}^{n\times n}(P)$. In particular, $\mathcal{H}_{\{0\}}^{n \times n}(P)$ is the set of all $n \times n$ Bohemian upper Hessenberg matrices with upper triangle entries from $P$ and subdiagonal entries equal to $1$ and $\mathcal{H}_{\{\pi\}}^{n \times n}(P)$ is when the subdiagonals entries are $-1$.
\end{definition}

\begin{theorem}
    \label{thm:charPolyRec1}
    \begin{equation}
        \label{eqn:thm1}
        Q_n(z) = zQ_{n-1}(z) - \sum_{k=1}^{n} s^{k-1} h_{n-k+1,n} Q_{n-k}(z)
    \end{equation}
    with the convention that $Q_0(z) = 1$ ($\mathbf{H}_0 = [\,]$, the empty matrix).
\end{theorem}

\begin{theorem}
    \label{thm:charPolyRec2}
    Expanding $Q_n(z)$ as
    \begin{equation}
        Q_n(z) = q_{n,n} z^n + q_{n,n-1} z^{n-1} + \cdots + q_{n,0},
    \end{equation}
    we can express the coefficients recursively by
    \begin{subequations}
    \label{eqn:thm2_eqn}
    \begin{align}
        q_{n,n} &= 1,\\
        q_{n,j} &= q_{n-1,j-1} - \sum_{k=1}^{n-j} s^{k-1} h_{n-k+1,n} q_{n-k,j} \quad\text{for}\quad 1 \le j \le n-1,\\
        q_{n,0} &= -\sum_{k=1}^n s^{k-1} h_{n-k+1,n} q_{n-k,0} \quad\text{for}\quad n>0,\quad\text{and}\\
        q_{0,0} &= 1\>.
    \end{align}
    \end{subequations}
\end{theorem}

\begin{definition}
    The \textit{characteristic height} of a matrix is the height of its 
    characteristic polynomial.
\end{definition}

\begin{proposition}
    \label{prop:negativeheight}
    For any matrix $\mathbf{A}$, $-\mathbf{A}$ has the same characteristic height as $\mathbf{A}$.
\end{proposition}

\begin{proposition}
    \label{prop:maxheight}
    The maximal characteristic height of $\mathbf{H}_n \in \mathcal{H}_{\{0,\pi\}}^{n \times n}(\{-1, 0, {+1}\})$ occurs when
    $s^{k-1} h_{i,i+k-1} = -1$ for $1 \le i \le n-k+1$ and $1 \le k \le n$.
\end{proposition}

\section{Upper Hessenberg Toeplitz Matrices}
For the remainder of the paper consider upper Hessenberg matrices with a Toeplitz structure of the form
\begin{equation}
    \mathbf{M}_n =
    \begin{bmatrix}
        t_1   & t_2    & t_3     & \cdots & t_n\\
        1   & t_1    & t_2     & \cdots & t_{n-1}\\
        0     & 1    & t_1     & \cdots & t_{n-2}\\
        \vdots & \ddots & \ddots & \ddots & \vdots\\
        0      & \cdots & 0       & 1    & t_1
    \end{bmatrix}
\end{equation}
with $t_k \in \{-1, 0, {+1}\}$ for $1 \le k \le n$.
Let
\begin{equation}
P_n(z) \equiv \det(z \mathbf{I} - \mathbf{M}_n) = \sum_{k=0}^n p_{n,k}z^k
\end{equation}
be the characteristic polynomial of $\mathbf{M}_n$ with $p_{n,n} = 1$.

% ---------------------------------------------------------------------------- %
% ---------------------------------------------------------------------------- %
\begin{proposition}
    \label{prop:charpolyrec1_UHT}
    The characteristic polynomial recurrence from Theorem~\ref{thm:charPolyRec1} can be written for upper Hessenberg Toeplitz matrices as
    \begin{equation}
        \label{eqn:thm1_UHT}
        P_n(z) = zP_{n-1}(z) - \sum_{k=1}^n t_k P_{n-k}(z)
    \end{equation}
        with the convention that $P_0(z) = 1$ ($\mathbf{M}_0 = [\,]$, the empty matrix).
\end{proposition}
\begin{proof}
    For a matrix $\mathbf{M}_n$, the entries at the $i$th row and the $i+k-1$-th column for $1 \le i \le n-k+1$ (i.e. the $k-1$-th diagonal) are all equal to $t_k$. In equation~\eqref{eqn:thm1}, we can replace $h_{n-k+1,n}$ with $t_k$ ($i=n-k+1$) recovering equation~\eqref{eqn:thm1_UHT}.
\end{proof}

% ---------------------------------------------------------------------------- %
% ---------------------------------------------------------------------------- %
\begin{proposition}
    \label{prop:charpolyrec2_UHT}
    The characteristic polynomial recurrence from Theorem~\ref{thm:charPolyRec2}  can be written for upper Hessenberg Toeplitz matrices as
    \begin{subequations}
    \label{eqn:thm2_cor_UHT}
    \begin{align}
        p_{n,n} &= 1,\\
        p_{n,j} &= p_{n-1,j-1} - \sum_{k=1}^{n-j} t_k p_{n-k,j} \quad\text{for}\quad 1 \le j \le n-1, \label{eqn:them2_cor_UHT_b}\\
        p_{n,0} &= -\sum_{k=1}^n t_k p_{n-k,0}, \,\text{and} \label{eqn:them2_cor_UHT_c}\\
        p_{0,0} &= 1 \label{eqn:them2_cor_UHT_d}\>.
    \end{align}
    \end{subequations}
\end{proposition}
\begin{proof}
    Performing the same replacement as above (a notational change), we recover equation~\eqref{eqn:thm2_cor_UHT}.
\end{proof}

% ---------------------------------------------------------------------------- %
% Proposition: p_{n,i} is a function of t_j for j <= n-i                       %
% ---------------------------------------------------------------------------- %
\begin{proposition}
    \label{prop:funof}
    $p_{n,i}$ is independent of $t_j$ for $j > n-i$.
\end{proposition}
\begin{proof}
    First, assume $p_{n,\ell}$ is a function of $t_1, \ldots, t_{n-\ell}$ for $\ell = i$ and all $n$. By Proposition~\ref{prop:charpolyrec2_UHT}
    \begin{equation}
        p_{n,\ell} = p_{n-1,\ell-1} - \sum_{k=1}^{n-\ell} t_k p_{n-k,\ell} \>.
    \end{equation}
    Isolating the $p_{n-1,\ell-1}$ term, we have
    \begin{equation}
        p_{n-1,\ell-1} = p_{n,\ell} + \sum_{k=1}^{n-\ell} t_k p_{n-k,\ell} 
    \end{equation}
    The first term, $p_{n,\ell}$, is a function of $t_1, \ldots, t_{n-\ell}$. Each term $t_k p_{n-k,\ell}$ in the sum is a function of $t_1, \ldots, t_{n-k-\ell}, t_k$. Taking $k=n-\ell$, we have the sum is a function of $t_1, \ldots, t_{n-\ell}$. Hence, $p_{n-1,\ell-1}$ is a function of $t_1, \ldots, t_{n-1-(\ell-1)} = t_{n-\ell}$.
    
    When $i = 0$, by Proposition~\ref{prop:charpolyrec2_UHT} we have
    \begin{equation}
        p_{n,0} = -\sum_{k=1}^n t_k p_{n-k,0}
    \end{equation}
    which is a function of $t_1, \ldots, t_n$.
\end{proof}

\begin{theorem}
    The set of characteristic polynomials for all matrices $\mathbf{M}_n$ with $t_k \in \{-1, 0, +1\}$ for $1 \le k \le n$ has cardinality $3^n$.
\end{theorem}
\begin{proof}
    Let
    \begin{equation}
    \mathbf{A}_n =
    \begin{bmatrix}
        a_1   & a_2    & a_3     & \cdots & a_n\\
        1   & a_1    & a_2     & \cdots & a_{n-1}\\
        0     & 1    & a_1     & \cdots & a_{n-2}\\
        \vdots & \ddots & \ddots & \ddots & \vdots\\
        0      & \cdots & 0       & 1    & a_1
    \end{bmatrix}
    \end{equation}
    with $a_k \in \{-1, 0, +1\}$ for $1 \le k \le n$.
    Let $R_n(z; a_1, \ldots, a_n)$ be the characteristic polynomial of $\mathbf{A}_n$.
    Assume $P_{\ell} = R_{\ell}$ for $\ell < n$. 
    By Proposition~\ref{prop:charpolyrec1_UHT}, for $\mathbf{A}_n$ and $\mathbf{M}_n$ to have the same characteristic polynomial we find
    \begin{equation}
    zP_{n-1} - \sum_{k=1}^n t_k P_{n-k} = zR_{n-1} - \sum_{k=1}^n a_k R_{n-k} \>.
    \end{equation}
    Since $P_{\ell} = R_{\ell}$ for all $\ell < n$, and the $\sum_{k=1}^n t_k P_{n-k}$ and $\sum_{k=1}^n t_k R_{n-k}$ terms are polynomials of degree $n-1$ in $z$, we find $P_n = R_n$ only when $t_k = a_k$ for all $1 \le k \le n$ (the $zP_{n-1}$ and $zR_{n-1}$ terms are the only terms of degree $n$ in $z$). Hence, for each combination of $t_k$, no other upper Hessenberg Toeplitz matrix with $t_k \in \{-1, 0, +1\}$ and subdiagonal $1$ has the same characteristic polynomial.
\end{proof}

% \begin{proposition}
%     The set of all characteristic polynomials of all matrices in $\mathcal{M}_n$ has cardinality $3^n$.
% \end{proposition}
% \begin{proof}
%     By Corollary~\ref{cor:onlysim}, each matrix $\mathbf{M}_n \in \mathcal{M}_n$ is similar to only one other matrix $\mathbf{M}_n^{+} \in \mathcal{M}_n$. Since $\mathbf{M}_n \ne \mathbf{M}_n^{+}$, and $\mathbf{M}_n$ shares a characteristic polynomial with $\mathbf{M}_n^{+}$; the cardinality of the set of characteristic polynomials must be half the cardinality of $\mathcal{M}_n$.
% \end{proof}

% ============================================================================ %
% Maximal Characteristic Height Upper Hessenberg Toeplitz Matrices             %
% ============================================================================ %
% Remove
\section{Maximal Characteristic Height Upper Hessenberg Toeplitz Matrices}
\begin{theorem}
    \label{thm:maxheight_UHT}
    The characteristic height of $\mathbf{M}_n$ is maximal when $t_k = -1$ for $1 \le k \le n$.
\end{theorem}
\begin{proof}
    Following from Proposition~\ref{prop:maxheight}, the entries in the $i$th row and $i+k-1$-th column for $1 \le i \le n-k+1$ correspond to $t_k$, after substituting $s=1$ we find
    $t_k = -1$ gives the maximal characteristic height.
\end{proof}

\begin{proposition}
    \label{prop:max_height_set}
    Let $F \subset \mathbb{R}$ be a closed and bounded set with $a = \min{F}$, $b = \max{F}$ and $\#F \ge 2$. Let $\mathbf{M}_n$ be upper Hessenberg Toeplitz with $t_k \in F$. If $|a| \ge |b|$, $\mathbf{M}_n$ attains maximal characteristic height for $t_k = a$ for all $1 \le k \le n$. If $|b| \ge |a|$, $\mathbf{M}_n$ attains maximal characteristic height for $t_k = a$ for $k$ even, and $t_k = b$ for $k$ odd.
\end{proposition}
\begin{proof}
    First, consider the case when $|a| \ge |b|$. Since $a < b$ we find $a < 0$. Let $\overline{t}_k = -t_k$. Writing Proposition~\ref{prop:negativeheight} in terms of $\overline{t}_k$ gives
    \begin{subequations}
    \label{eq:polyrec_neg_t}
    \begin{align}
        p_{n,n} &= 1,\\
        p_{n,j} &= p_{n-1,j-1} + \sum_{k=1}^{n-j} \overline{t}_k p_{n-k,j} \quad\text{for}\quad 1 \le j \le n-1,\\
        p_{n,0} &= \sum_{k=1}^n \overline{t}_k p_{n-k,0}, \,\text{and}\\
        p_{0,0} &= 1\>.
    \end{align}
    \end{subequations}
    If all $\overline{t}_k$ are positive then $p_{n,j}$ must be positive for all $n$ and $j$. Hence, the maximal characteristic height is attained when $\overline{t}_k$ is maximal, or equivalently $t_k$ is minimal and negative. Thus $t_k = \min{F} = a$ gives maximal characteristic height.

    Next, consider when $|b| \ge |a|$. Since $a < b$ we find $b > 0$. By Proposition~\ref{prop:negativeheight} we know that the characteristic height of $\mathbf{M}_n$ is equal to the characteristic height of $-\mathbf{M}_n$. Rewriting Proposition~\ref{prop:charpolyrec2_UHT} for $-\mathbf{M}_n$ by substituting $p_{n,j}$ with $(-1)^{n-j} p_{n,j}$ we find the recurrence for the characteristic polynomial of $-\mathbf{M}_n$:
    \begin{subequations}
    \label{eq:polyrec_neg_mat}
    \begin{align}
        p_{n,n} &= 1,\\
        p_{n,j} &= p_{n-1,j-1} + \sum_{k=1}^{n-j} (-1)^{k-1} t_k p_{n-k,j} \quad\text{for}\quad 1 \le j \le n-1,\\
        p_{n,0} &= \sum_{k=1}^n (-1)^{k-1} t_k p_{n-k,0}, \,\text{and}\\
        p_{0,0} &= 1\>.
    \end{align}
    \end{subequations}
    Separating out the even and odd values of $k$ in the sums we can write the recurrence as
    \begin{subequations}
    \label{eq:polyrec2_odd_even}
    \begin{align}
        p_{n,n} &= 1,\\
        p_{n,j} &= p_{n-1,j-1} + \sum_{k \text{ odd}}^{n-j} t_k p_{n-k,j} - \sum_{k \text{ even}}^{n-j} t_k p_{n-k,j} \quad\text{for}\quad 1 \le j \le n-1,\\
        p_{n,0} &= \sum_{k \text{ odd}}^n t_k p_{n-k,0} - \sum_{k \text{ even}}^n t_k p_{n-k,0}, \,\text{and}\\
        p_{0,0} &= 1\>.
    \end{align}
    \end{subequations}
    The odd sums are maximal for $t_k = \max{F} = b$ and the even sums are maximal for $t_k = \min{F} = a$. Hence, the maximal characteristic height is attained for $t_k = b$ when $k$ is odd, and $t_k = a$ when $k$ is even.
    
    When $|a| = |b|$, equations~\eqref{eq:polyrec_neg_t} and~\eqref{eq:polyrec2_odd_even} are equivalent and the maximal height is attained both when $t_k = b$ for all $k$, and $t_k = b$ for $k$ odd and $t_k = a$ for $k$ even.
\end{proof}

\begin{proposition}
    \label{prop:maxheight_UHT2}
    $\mathbf{M}_n$ also attains maximal characteristic height when $t_k = (-1)^{k-1}$ for $1 \le k \le n$.
\end{proposition}
\begin{proof}
    By Proposition~\ref{prop:max_height_set}, we have $F = \{-1, 0, {+1}\}$ with $a = -1$, and $b = +1$. Thus $\mathbf{M}_n$ is also of maximal characteristic height for $t_k = b = +1$ for odd values of $k$, and $t_k = a = -1$ for even values of $k$.
\end{proof}

% \begin{remark}
%     Propositions~\ref{thm:maxheight_UHT} and~\ref{prop:maxheight_UHT2} give four solutions where $\mathbf{M}_n$ attains maximal characteristic height:
%     \begin{enumerate}
%         \item $s = 1$ and $t_k = -1$ for $1 \le k \le n$,
%         \item $s = -1$ and $t_k = (-1)^k$ for $1 \le k \le n$,
%         \item $s = -1$ and $t_k = 1$ for $1 \le k \le n$, and
%         \item $s = 1$ and $t_k = (-1)^{k+1}$ for $1 \le k \le n$.
%     \end{enumerate}
% \end{remark}

\begin{proposition}
    The maximum characteristic height grows at least exponentially in $n$.
\end{proposition}
\begin{proof}
    When $t_k = -1$, the characteristic height is maximal by Theorem~\ref{thm:maxheight_UHT}.
    Equation~\eqref{eqn:them2_cor_UHT_c} from Proposition~\ref{prop:charpolyrec2_UHT} reduces to
    \begin{equation}
        p_{n,0} = \sum_{k=1}^n p_{n-k,0} = 2^{n-1}
    \end{equation}
    for $n \ge 1$ with $p_{0,0} = 1$ by equation~\eqref{eqn:them2_cor_UHT_d}. Thus, the maximal characteristic height must grow at least exponentially in $n$.
\end{proof}

\begin{conjecture}
    The maximum characteristic height approaches $C (1 + \varphi)^n$ as $n \to \infty$ for some constant $C$ where $\varphi$ is the golden ratio.
\end{conjecture}
\begin{remark}
    This limit is illustrated in Figure~\ref{fig:maxheightratiolog}, motivating this conjecture.
\end{remark}
\begin{figure}[h]
    \centering
    \includegraphics[width=\textwidth]{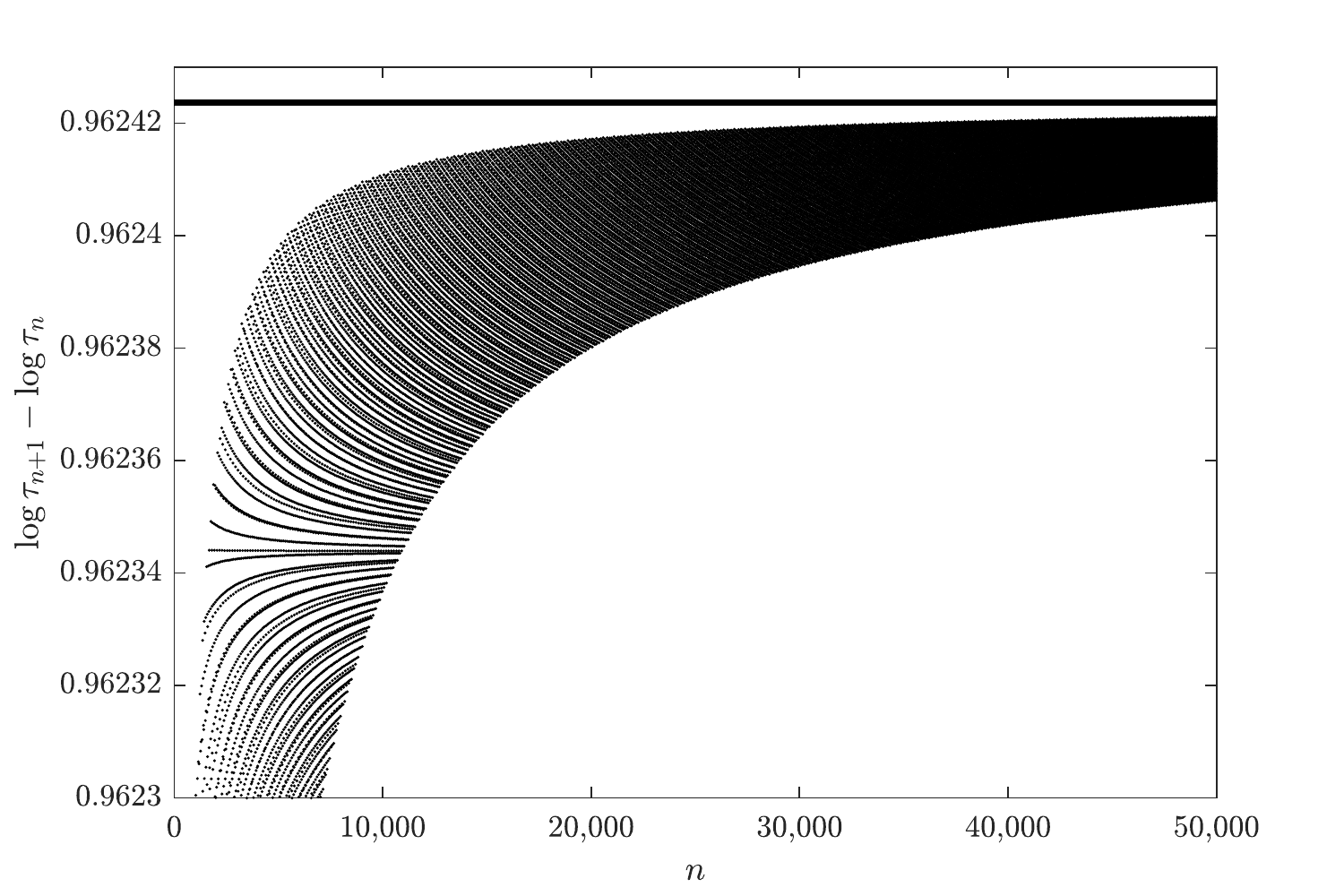}
    \caption{The points are $\log{\tau_{n+1}} - \log{\tau_n}$ for $n$ from 0 to 50,000 where $\tau_n$ is the maximal characteristic height of $\mathbf{M}_n$ (i.e. when $t_k = -1$, for example). The solid line is $\log(1 + \varphi)$ where $\varphi$ is the golden ratio.}
    \label{fig:maxheightratiolog}
\end{figure}

\begin{proposition}
    \label{prop:heightcoeffs}
    Let $\overline{\mathbf{M}}_n$ be of maximal characteristic height and let $\mu_n$ be the degree of the term of the characteristic polynomial of  $\overline{\mathbf{M}}_n$ corresponding to the height. The characteristic height of $\overline{\mathbf{M}}_n$ is independent of $t_j$ for $j > n - \mu_n$.
\end{proposition}
\begin{proof}
    Let $P_n$ be the characteristic polynomial of $\overline{\mathbf{M}}_n$.
    By Proposition~\ref{prop:funof}, $p_{n,\mu_n}$ is independent of $t_j$ for $j > n -  \mu_n$. Thus, $t_j$ for $j > n -  \mu_n$ only affects $p_{n,k}$ for $k <  \mu_n$. Since $\overline{\mathbf{M}}_n$ is of maximal height, $|p_{n,k}| \le |p_{n,  \mu_n}|$ for $k < \mu_n$ for all $t_j\in \{-1, 0, +1\}$ with $j > n -  \mu_n$.
\end{proof}

\begin{table}[h!]
  \begin{center}
    \begin{tabular}{ccc}
    \hline
      $n$ & $\mu_n$ & $\tau_n$\\ \hline
      2 & 1 & 2\\
      3 & 1 & 5\\
      4 & 1 & 12\\
      5 & 1 & 27\\
      6 & 2 & 66\\
      7 & 2 & 168\\
      8 & 2 & 416\\
      9 & 2 & 1,008\\
      10 & 3 & 2,528\\ \hline
    \end{tabular}
  \end{center}
  \label{tab:max_char_height}
  \caption{Maximum height $\tau_n$ and degree of term of characteristic polynomial corresponding to maximum height $\mu_n$ upper Hessenberg Toeplitz matrices for $n$ from 2 to 10.}
\end{table}

\begin{proposition}
    For fixed $n$, $\mu_n$ is the same for all matrices $\overline{\mathbf{M}}_n$ of maximal characteristic height.
\end{proposition}
\begin{proof}
The characteristic polynomial of $\mathbf{M}_n$ when $t_k = -1$ has the same coefficients as the characteristic polynomial of $\mathbf{M}_n$ for $t_k = (-1)^{k-1}$ up to a sign change.
    % When $(-1)^k t_k = 1$, the characteristic polynomial of $\overline{\mathbf{M}}_n$ is the same for both solutions.
    % When $s^{k-1}t_k = (-1)^{k-1}$, the coefficients of the characteristic polynomial of $\overline{\mathbf{M}}_n$ are the same as when $s^{k-1}t_k = -1$ up to a sign change.
    By Proposition~\ref{prop:heightcoeffs}, changing any of the entries $t_j$ of $\overline{\mathbf{M}}_n$ for $j > n-\mu_n$ does not affect the value of $\mu_n$. Therefore $\mu_n$ is fixed.
\end{proof}

\begin{theorem}
\label{thm:count_max_height_uht}
The number of upper Hessenberg Toeplitz matrices of dimension $n$ with $t_k \in \{-1, 0, {+1}\}$ for $1 \le k \le n$ of maximal characteristic height is $2\cdot3^{\mu_n}$.
\end{theorem}
\begin{proof}
    By Theorem~\ref{thm:maxheight_UHT} and Proposition~\ref{prop:maxheight_UHT2}, there are two matrices that attain maximal characteristic height. By Proposition~\ref{prop:heightcoeffs},
    any combination of $t_j \in \{-1, 0, {+1}\}$ for $j > n - \mu_n$ will not affect the characteristic height. Thus there are $3^{\mu_n}$ combinations of $t_j$ that result in the same characteristic height for each of the two choices of $t_k$ that give maximal characteristic height.
\end{proof}

\begin{remark}
We have found that $\mu_n$ remains constant for 3 or 4 subsequent values of $n$ followed by an increment by 1. We have verified this pattern experimentally up to degree 50,000. Figure~\ref{fig:argmaxidx} shows the pattern for matrix dimension up to 100.
\end{remark}

\begin{figure}[h]
    % See PlotArgmaxIdx.m and PlotArgMaxIdx.mw
    \centering
    \includegraphics[width=\textwidth]{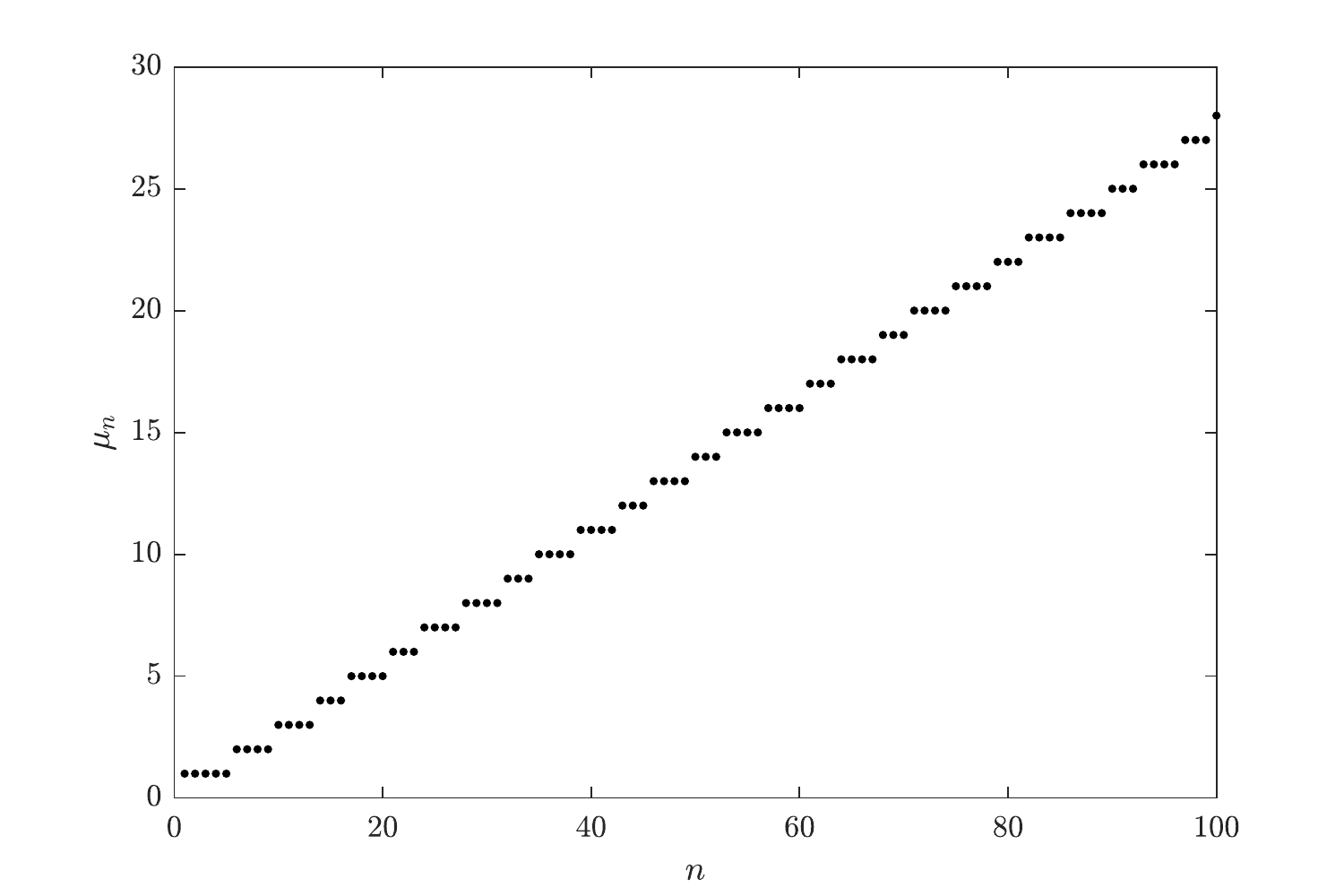}
    \caption{Degree of the term corresponding to the height of the characteristic polynomial of an $n \times n$ upper Hessenberg Toeplitz matrix of maximal characteristic height.}
    \label{fig:argmaxidx}
\end{figure}
    
    % \begin{figure}
    %     \centering
    %     \includegraphics[width=6in]{Figures/MaxHeight.pdf}
    %     \caption{Maximum characteristic height of $n \times n$ upper Hessenberg Toeplitz matrices.}
    % \label{fig:maxheight}
    % \end{figure}

\begin{remark} 
    The sequence $\mu_{n+1} - \mu_n$ is nearly equivalent to the sequence for the generalized Fibonacci word $f^{[3]}$
    \begin{equation}
        a(n) = \floor*{\frac{n+2}{\varphi+2}} - \floor*{\frac{n+1}{\varphi+2}}
    \end{equation}
    (\href{http://oeis.org/A221150}{A221150} on the OEIS). We have found that up to at least degree 50,000, $\mu_{n+1} - \mu_n = a(n+326)$ except when $n \in \{0, 2, 24,148, 24,149\}$.
\end{remark}

\begin{remark}
    The sequence $\mu_n$ is nearly equivalent to the sequence
    \begin{equation}
        \floor*{\frac{n+327}{\varphi+2}} - 90
    \end{equation}
    for $n > 2$. The two sequences are equal for all values up to $n = 50,000$ except when $n = 24,149$.
\end{remark}

The sequences presented in the previous remarks are examples of \textit{high-precision fraud}~\cite{borwein1992strange} requiring evaluation up to dimension 25,000 and nearly 25,000 digits of precision to identity.

% ============================================================================ %
% Maximal Height Characteristic Polynomials                                    %
% ============================================================================ %
\section{Maximal Height Characteristic Polynomials}
In this section we restrict our analysis to specific upper Hessenberg Toeplitz matrices of maximal characteristic height, that is $t_k = -1$ for all $k$. We denote a dimension $n$ matrix of this form by $\widetilde{\mathbf{M}}_n$.
$\widetilde{\mathbf{M}}_n$ is of maximal height by Proposition~\ref{prop:maxheight_UHT2}.

\begin{proposition}
    The characteristic polynomial of $\widetilde{\mathbf{M}}_n$ is of the form
    \begin{equation}
        P_n = z^n + p_{n,n-1}z^{n-1} + \cdots + p_{n,0}
    \end{equation}
    where $p_{n,j}$ is positive for all $n$ and $j$.
\end{proposition}
\begin{proof}
    When $t_k = -1$ for $1 \le k \le n$, Proposition~\ref{prop:charpolyrec2_UHT} reduces to
    \begin{subequations}
    \label{eqn:max_height_char_poly_rec}
    \begin{align}
        p_{n,n} &= 1,\\
        p_{n,j} &= p_{n-1,j-1} + \sum_{k=1}^{n-j} p_{n-k,j}\quad\text{for}\quad 1 \le j \le n-1, \label{eqn:max_height_char_poly_rec_b}\\
        p_{n,0} &= \sum_{k=1}^n p_{n-k,0},
 \,\text{and} \label{eqn:max_height_char_poly_rec_c}\\
        p_{0,0} &= 1 \label{eqn:max_height_char_poly_rec_d}\>.
    \end{align}
    \end{subequations}
    Since $p_{0,0}$ is positive, and all coefficients in the above equations are positive, $p_{n,j}$ must be positive for all $n$ and $j$.
\end{proof}

\begin{proposition}
    \label{prop:genfun}
    The generating function of the sequence $(p_{i,i}, p_{i+1,i}, \ldots)$ for all $i \ge 0$ is
    \begin{equation}
        G_{i}(x) = \bigg (\frac{1-x}{1-2x} \bigg)^{i+1} \>.
    \end{equation}
\end{proposition}
\begin{proof}
    First we will prove the $i=0$ case. Let 
    \begin{equation}
        G_0(x) = \sum_{\ell=0}^{\infty} p_{\ell,0} x^\ell \>.
    \end{equation}
    Then,
    \begin{equation}
        (1-2x) G_0(x) = p_{0,0} + \sum_{\ell=1}^{\infty} (p_{\ell,0} - 2p_{\ell-1,0}) x^\ell \>.
    \end{equation}
    From equation~\eqref{eqn:max_height_char_poly_rec_c},
    \begin{align}
        (1-2x) G_0(x) &= p_{0,0} + (p_{1,0} - 2p_{0,0}) x + \sum_{\ell=2}^{\infty} (p_{\ell,0} - 2p_{\ell-1,0})x^\ell\\
        &= p_{0,0} + (p_{1,0} - 2p_{0,0}) x + \sum_{\ell=2}^{\infty} \bigg ( \sum_{k=1}^{\ell} p_{\ell-k,0} - 2 \sum_{k=1}^{\ell-1} p_{\ell-1-k,0} \bigg ) x^{\ell} \\
        &= p_{0,0} + (p_{1,0} - 2p_{0,0}) x + \sum_{\ell=2}^{\infty} \bigg ( \sum_{k=1}^{\ell} p_{\ell-k,0} - 2 \sum_{k=2}^{\ell} p_{\ell-k,0} \bigg ) x^{\ell}\\
        &= p_{0,0} + (p_{1,0} - 2p_{0,0}) x + \sum_{\ell=2}^{\infty} \bigg ( p_{\ell-1,0} - \sum_{k=2}^{\ell} p_{\ell-k,0} \bigg ) x^{\ell} \>.
    \end{align}
    Since $p_{0,0} = p_{1,0} = 1$,
    \begin{align}
        (1-2x) G_0(x) &= 1 - x + \sum_{\ell = 2}^{\infty} \bigg( p_{\ell-1,0} - \sum_{k=1} ^{\ell-1} p_{\ell-1-k,0} \bigg)x^{\ell} \\
        &= 1-x \>.
    \end{align}
    Therefore
    \begin{equation}
        G_0(x) = \frac{1-x}{1-2x} \>.
    \end{equation}
    
    Next we prove the general case for $i > 0$. Assume inductively that
    \begin{equation}
        G_i(x) = \bigg ( \frac{1-x}{1-2x} \bigg )^{i+1} = \sum_{\ell=0}^{\infty} p_{i+\ell,i} x^\ell \>.
    \end{equation}
    
    \begin{align*}
        \sum_{\ell = 0}^{\infty} p_{i+\ell+1,i+1}x^{\ell} &= \bigg( \frac{1-2x}{1-2x}\bigg )\sum_{\ell = 0}^{\infty} p_{i+\ell+1,i+1}x^{\ell}\\
        &= \bigg ( \frac{1}{1-2x} \bigg ) \bigg [ \sum_{\ell=0}^{\infty} p_{i+\ell+1,i+1} x^{\ell} - 2x \sum_{\ell = 0}^{\infty} p_{i+\ell+1,i+1}x^{\ell} \bigg ]\\
        &= \bigg ( \frac{1}{1-2x} \bigg ) \bigg [ p_{i+1,i+1} + \sum_{\ell=1}^{\infty} (p_{i+\ell+1,i+1} - 2p_{i+\ell,i+1}) x^{\ell} \bigg ] \shortintertext{Because $p_{i+1,i+1} = 1 = p_{i,i}$}\\
        &= \bigg ( \frac{1}{1-2x} \bigg ) \bigg [ p_{i,i} + \sum_{\ell=1}^{\infty} (p_{i+\ell+1,i+1} - 2p_{i+\ell,i+1}) x^{\ell} \bigg ]\\
        &= \bigg ( \frac{1}{1-2x} \bigg ) \bigg [ p_{i,i} + \sum_{\ell=1}^{\infty} \bigg ( p_{i+\ell+1,i+1} - p_{i+\ell,i+1} - p_{i+\ell,i+1} \bigg ) x^{\ell} \bigg ]\\
        &= \bigg ( \frac{1}{1-2x} \bigg ) \bigg [ p_{i,i} + \sum_{\ell=1}^{\infty} \bigg ( p_{i+\ell+1,i+1} - p_{i+\ell,i+1} - \sum_{k=0}^{\ell-1} p_{i+\ell-k,i+1} + \sum_{k=1}^{\ell-1} p_{i+\ell-k,i+1} \bigg ) x^{\ell} \bigg ]\\
        &= \bigg ( \frac{1}{1-2x} \bigg ) \bigg [ p_{i,i} + \sum_{\ell=1}^{\infty} \bigg ( p_{i+\ell+1,i+1} - p_{i+\ell,i+1} - \sum_{k=0}^{\ell-1} p_{i+\ell-k,i+1} + \sum_{k=1}^{\ell-1} p_{i+\ell-k,i+1} \bigg ) x^{\ell} \bigg ]\\
        &= \bigg ( \frac{1}{1-2x} \bigg ) \Bigg [ p_{i,i} + \sum_{\ell=1}^{\infty} \Bigg ( \bigg ( p_{i+\ell+1,i+1} - \sum_{k=1}^{\ell} p_{i+\ell+1-k,i+1} \bigg ) - \bigg ( p_{i+\ell,i+1} - \sum_{k=1}^{\ell-1} p_{i+\ell-k,i+1} \bigg ) \Bigg) x^{\ell} \Bigg ] \>.
    \end{align*}
    Rewriting equation~\eqref{eqn:max_height_char_poly_rec_b} as
    \begin{equation}
        p_{n,j} = p_{n+1,j+1} - \sum_{k=1}^{n-j} p_{n+1-k,j+1},
    \end{equation}
    we find
    \begin{align*}
        \sum_{\ell = 0}^{\infty} p_{i+\ell+1,i+1}x^{\ell} &= \bigg ( \frac{1}{1-2x} \bigg ) \bigg [ p_{i,i} + \sum_{\ell=1}^{\infty} (p_{i+\ell,i} - p_{i+\ell-1,i}) x^{\ell} \bigg ]\\
        &= \bigg ( \frac{1}{1-2x} \bigg ) \bigg [ \sum_{\ell=0}^{\infty} p_{i+\ell,i} x^{\ell} - \sum_{\ell=1}^{\infty} p_{i+\ell-1,i} x^{\ell} \bigg ]\\
        &= \bigg ( \frac{1}{1-2x} \bigg ) \bigg [ \sum_{\ell=0}^{\infty} p_{i+\ell,i}x^{\ell} - \sum_{\ell=0}^{\infty} p_{i+\ell,i} x^{\ell+1} \bigg ]\\
        &= \bigg ( \frac{1-x}{1-2x} \bigg ) \sum_{\ell=0}^{\infty} p_{i+\ell,i} x^{\ell}\\
        &= \bigg ( \frac{1-x}{1-2x} \bigg )^{i+2}
    \end{align*}
    \end{proof}

\begin{proposition}
    The coefficients $p_{n,k}$ are given by the OEIS sequence \href{http://oeis.org/A105306}{A105306} for the ``number of directed column-convex polynomials of area $n$, having the top of the right-most column at height $k$.'' We have $p_{n,k} = T_{n+1,k+1}$ where
    \begin{equation}
        T_{n,k} = 
        \begin{cases}
            \displaystyle\sum_{j=0}^{n-k-1} \binom{k+j}{k-1} \binom{n-k-1}{j} & \text{if } k < n\\
            \hfil 1 & \text{if } k = n
        \end{cases}
    \end{equation}
    Maple ``simplifies'' this to
    \begin{equation}
      T_{n,k} =
      \begin{cases}
          kF \!
          \left(
              \begin{array}{c|c}
                  k+1, k+1-n & \multirow{2}{*}{$-1$} \\
                  2 &
              \end{array}
          \right) & \text{if } n \ne k\\
          \hfil 1 & \text{if } n = k
        \end{cases}
    \end{equation}
    where $F(\cdot)$ is the hypergeometric function defined as
    \begin{equation}
        F \!
          \left(
              \begin{array}{c|c}
                  a, b & \multirow{2}{*}{$z$} \\
                  c &
              \end{array}
          \right)
          =
          \sum_{n=0}^{\infty} \frac{a^{\bar{n}} b^{\bar{n}}}{c^{\bar{n}}} \frac{z^n}{n!}
    \end{equation}
    where $q^{\bar{n}}$ is $q\cdot(q+1)\cdots(q + n - 1)$.
\end{proposition}
\begin{proof}
    We will show that
    \begin{equation}
        p_{i+n,i} = T_{n+i+1,i+1} =
        \begin{cases}
            \displaystyle \sum_{j=0}^{n-1} {i+j+1 \choose i} {n-1 \choose j} & \text{if } n > 0\\
            \hfil 1 & \text{if } n = 0 \>.
        \end{cases}
    \end{equation}

    By Proposition~\ref{prop:genfun}
    \begin{equation}
        p_{i+n,i} = \frac{1}{n!} \frac{\mathrm{d}^n}{\mathrm{d}x^n} G_i(x) \Bigr|_{x=0}
    \end{equation}
    where
    \begin{equation}
        G_i(x) = \bigg ( \frac{1-x}{1-2x} \bigg )^{i+1} = f_i(g(x))
    \end{equation}
    with
    \begin{align}
        f_i(x) &= x^{i+1}, \text{ and}\\
        g(x) &= \frac{1-x}{1-2x} = \frac{1}{1-2x} - \frac{x}{1-2x}\>.
    \end{align}
    Differentiating $f_i(x)$ and $g(x)$ with respect to $x$,
    \begin{align}
        \frac{\mathrm{d}^n}{\mathrm{d}x^n} f_i(x) &=
        \begin{cases}
            (i+1)(i)\cdots(i-n+2) x^{i+1-n} &\text{for}\quad n \le i+1\\
            \hfil 0 &\text{for}\quad n > i+1
        \end{cases}\\
        &= {i+1 \choose n} n! x^{i+1-n}
    \end{align}
    and
    \begin{align}
        \frac{\mathrm{d}^n}{\mathrm{d}x^n} g(x) &= \frac{\mathrm{d}^n}{\mathrm{d}x^n} \frac{1}{1-2x} + \frac{\mathrm{d}^n}{\mathrm{d}x^n} \frac{x}{1-2x}\\
        &= \frac{2^n n!}{(1-2x)^{n+1}} + \frac{2^{n-1} n!}{(1-2x)^n} + \frac{2^n n! \,x}{(1-2x)^{n+1}}\\
        &= \frac{2^n n! (1-x)}{(1-2x)^{n+1}} - \frac{2^{n-1}n!}{(1-2x)^n}
    \end{align}
    with
    \begin{equation}
    \frac{\mathrm{d}^n}{\mathrm{d}x^n} g(x) \Bigr|_{x=0} =
    \begin{cases}
        n!\,2^{n-1} &\text{for}\quad n > 0\\
        1 &\text{for}\quad n = 0 \>.
    \end{cases}
    \end{equation}
    
    When $n=0$,
    \begin{equation}
        p_{i+n,i} = p_{i,i} = G_i(0) = 1 \>.
    \end{equation}
    For $n > 0$, by Fa\`a di Bruno's formula we have
    \begin{align}
        \frac{\mathrm{d}^n}{\mathrm{d}x^n} G_i(x) &= \frac{\mathrm{d}^n}{\mathrm{d}x^n} f_i(g(x)) \\
        &= \sum_{k=1}^n f_i^{(k)}\big(g(x)\big) B_{n,k} (g'(x), g''(x), \ldots, g^{(n-k+1)}(x))
    \end{align}
    and therefore
    \begin{align}
        \frac{\mathrm{d}^n}{\mathrm{d}x^n} G_i(x) \Bigr|_{x=0} &= \sum_{k=1}^n f_i^{(k)} \big ( g(0) \big) B_{n,k} (g'(0), g''(0), \ldots, g^{(n-k+1)}(0))\\
        &= \sum_{k=1}^n f_i^{(k)}(1) B_{n,k}(1, 4, 24, \ldots, (n-k+1)! 2^{n-k}) \>.
    \end{align}
    
    By Theorem 6 of~\cite{abbas2005new},
    \begin{align}
         B_{n,k}(1, 4, 24, \ldots, (n-k+1)! 2^{n-k}) &= B_{n,k}(q_0(1), q_1(2), \ldots, q_{n-k}(n-k+1)) \\
         &= {n-1 \choose k-1} \frac{n!}{k!} 2^{n-k}
    \end{align}
    because the function
    \begin{equation}
        q_n(x) = \frac{x!}{(x-n)!} 2^n
    \end{equation}
    satisfies
    \begin{equation}
        q_n(x+y) = \sum_{k=0}^n {n \choose k} q_k(y) q_{n-k}(x) \>.
    \end{equation}
    
    Returning to the proof,
    \begin{align}
        p_{i+n,i} &= \frac{1}{n!} \frac{\mathrm{d}^n}{\mathrm{d}x^n} G_i(x)\Bigr|_{x=0} \\
        &= \frac{1}{n!} \sum_{k=1}^n {i+1 \choose k} {n-1 \choose k-1} k! \frac{n!}{k!} 2^{n-k} \\
        &= \sum_{k=1}^n {i+1 \choose k} {n-1 \choose k-1} 2^{n-k} \\
        &= \sum_{k=0}^{n-1} {i+1 \choose k+1} {n-1 \choose k} 2^{n-k-1} \\
        &= \sum_{k=0}^{n-1} {i+1 \choose k+1} {n-1 \choose k} \sum_{j = 0}^{n-k-1} {n-k-1 \choose j}\\
        &= \sum_{k=0}^{n-1} \sum_{j=0}^{n-k-1} {n-1 \choose k} {i+1 \choose k+1} {n-k-1 \choose j}\\
        &= \sum_{j = 0}^{n-1} \sum_{k = 0}^{n - j - 1} {n-1 \choose k} {i+1 \choose k+1} {n-k-1 \choose j}\\
        &= \sum_{j = 0}^{n-1} \sum_{k=0}^{j} {n-1 \choose k} {i+1 \choose k+1} {n-k-1 \choose n-j - 1}\\
        &= \sum_{j = 0}^{n-1} \sum_{k=0}^{j} {n-1 \choose n-j - 1} {j \choose k} {i+1 \choose k+1}\\
        &= \sum_{j = 0}^{n-1} {n-1 \choose j} \sum_{k=0}^{j} {j \choose k} {i+1 \choose k+1} \\
        &= \sum_{j = 0}^{n-1} {n-1 \choose j} {i+j+1 \choose j + 1} \\
        &= \sum_{j = 0}^{n-1} {n-1 \choose j} {i+j+1 \choose i}
    \end{align}
\end{proof}

% ------------------------- %
% A 2D recurrence           %
% ------------------------- %
\begin{proposition}
    The characteristic polynomial of $\widetilde{\mathbf{M}}_n$ is
    \begin{equation*}
    P_n(z)
    = \sum_{\ell = 0}^{\lfloor \sfrac{n}{2} \rfloor}
    {n \choose 2\ell}
    \bigg(\frac{z}{2} + 1\bigg)^{n - 2\ell}\bigg(1 + \frac{z^2}{4}\bigg)^{\ell} + \dfrac{z}{2}\sum_{\ell = 0}^{\lfloor \frac{n-1}{2}\rfloor}
    {n \choose 2\ell + 1}
    \bigg(\frac{z}{2} + 1\bigg)^{n - 2\ell - 1}\bigg(1 + \frac{z^2}{4}\bigg)^{\ell} \>.
    \end{equation*}
\end{proposition}
This proposition can be proved in several ways. We choose below to think of $z \in \mathbb{C} \setminus \{\pm 2i\}$, for a reason that will become clear. Since the end result is a polynomial in $z$, proving the formula for $z\neq \pm 2i$ will recover the exceptional cases by continuity.

Another equally valid approach would be to think of $z$ as being transcendental and noting that the characteristic polynomial of $\widetilde{\mathbf{M}}_{n}$ has integer coefficients.
\begin{proof}
From Proposition~\ref{prop:charpolyrec1_UHT} 
\begin{align}
    P_n(z) &= zP_{n-1}(z) - \sum_{k=1}^n t_k P_{n-k}(z)\\
    &= zP_{n-1}(z) - \sum_{k=0}^{n-1} t_{n-k} P_k(z) \>.
\end{align}

If $t_k = -1$ for $1 \le k \le n$,
\begin{equation}
    P_n(z) = zP_{n-1}(z) + \sum_{k=0}^{n-1} P_k(z) \>.
\end{equation}
Let $T_j(z) = \sum_{k=0}^{j} P_k(z)$. $T_n(z) = T_{n-1}(z) + P_{n}(z)$, so
\begin{align}
    P_n(z) &= zP_{n-1}(z) + T_{n-1}(z) \\
    T_n(z) &= zP_{n-1}(z) + 2T_{n-1}(z)
\end{align}
or
\begin{align}
    \left[
        \begin{array}{c}
            P_n(z) \\
            T_n(z)
        \end{array}
    \right]
    &=
    \left[
        \begin{array}{cc}
            z & 1 \\
            z & 2
        \end{array}
    \right]^n
    \left[
        \begin{array}{c}
            P_0(z) \\
            T_0(z)
        \end{array}
    \right] \\
    &= 
    \left[
        \begin{array}{cc}
            z & 1 \\
            z & 2
        \end{array}
    \right]^n 
    \left[
        \begin{array}{c}
            1 \\
            1
        \end{array}
    \right]
\end{align}
since $P_0(z) = 1$ and $T_0(z) = \sum_{j = 0}^0 P_0(z) = 1$. The eigenvalues of this matrix are
\begin{align}
    \lambda_{+} &= 1 + \frac{z}{2} + \Delta \\
    \lambda_{-} &= 1 + \frac{z}{2} - \Delta \\
    \Delta &= \sqrt{1 + \sfrac{z^2}{4}} \>.
\end{align}
If $z = \pm2i$ the eigenvalues are multiple and our approach would have to be modified. We ignore this and recover the true result at the end. The eigenvectors are
\begin{equation}
    \mathbf{V} =
    \left[
        \begin{array}{cc}
            1 & 1 \\
            1 - \frac{z}{2} + \Delta & 1 - \frac{z}{2} - \Delta
        \end{array}
    \right]
\end{equation}
and
\begin{equation}
    \mathbf{V}^{-1} = \dfrac{-1}{2\Delta}
    \left[
        \begin{array}{cc}
             1 - \frac{z}{2} - \Delta & -1 \\
            -1 + \frac{z}{2} - \Delta & 1
        \end{array}
    \right]
\end{equation}
hence
\begin{equation}
    \mathbf{V}^{-1}
    \left[
        \begin{array}{c}
            1 \\
            1
        \end{array}
    \right]
    =
    \dfrac{-1}{2\Delta}
    \left[
        \begin{array}{c}
            \frac{-z}{2} - \Delta \\
            \frac{z}{2} - \Delta
        \end{array}
    \right]
    =
    \left[
        \begin{array}{c}
            \frac{1}{2} + \frac{z}{4\Delta} \\
            \frac{1}{2} - \frac{z}{4\Delta}
        \end{array}
    \right] \>.
\end{equation}
Therefore
\begin{equation}
    \left[
        \begin{array}{c}
            P_n(z) \\
            T_n(z)
        \end{array}
    \right]
    =
    \left[
        \begin{array}{cc}
            1 & 1 \\
            1 - \frac{z}{2} + \Delta & 1 - \frac{z}{2} - \Delta
        \end{array}
    \right]
    \left[
        \begin{array}{c}
            \lambda_{+}^{n}\left(\frac{1}{2} + \frac{z}{4\Delta}\right) \\
            \lambda_{-}^{n}\left(\frac{1}{2} - \frac{z}{4\Delta}\right)
        \end{array}
    \right]
\end{equation}
and in particular
\begin{equation}
    P_n(z) = \lambda_{+}^{n}\bigg(\frac{1}{2} + \frac{z}{4\Delta}\bigg) + \lambda_{-}^{n}\bigg(\frac{1}{2} - \frac{z}{4\Delta}\bigg) \>.
\end{equation}
Now
\begin{align}
    \lambda_{+}^n &= \bigg(\frac{z}{2} + 1 + \Delta \bigg)^n\\
    &= \sum_{k=0}^n {n \choose k} \bigg( \frac{z}{2} + 1 \bigg) \Delta^k
\end{align}
and
\begin{align}
    \lambda_{-}^n &= \bigg(\frac{z}{2} + 1 - \Delta \bigg)^n\\
    &= \sum_{k=0}^n {n \choose k} \bigg( \frac{z}{2} + 1 \bigg) \left(-\Delta\right)^k \>.
\end{align}
\begin{multline}
    \therefore P_n(z) = \sum_{k=0}^n {n \choose k} \bigg ( \frac{z}{2} + 1 \bigg )^{n-k} \bigg ( \frac{1}{2} \Delta^k + \frac{1}{2} (-\Delta)^k \bigg )\\+
    \frac{z}{4\Delta} \sum_{k=0}^n {n \choose k} \bigg ( \frac{z}{2}+1 \bigg)^{n-k} \left ( \Delta^k - (-\Delta)^k \right ) \>.
\end{multline}
Every odd term drops out of the first, and every even out of the second.
\begin{align*}
    \therefore P_n(z) &= \sum_{\substack{k = 0 \\ k \text{ even}}}^{n} {n \choose k} \bigg ( \frac{z}{2} + 1 \bigg )^{n-k} \Delta^k + \dfrac{z}{4\Delta}\sum_{\substack{k = 0 \\ k \text{ odd}}}^{n} {n \choose k} \bigg(\frac{z}{2} + 1\bigg)^k \cdot2\Delta^k \\
    &= \sum_{\ell = 0}^{\lfloor \sfrac{n}{2} \rfloor}
    {n \choose 2\ell}
    \bigg(\frac{z}{2} + 1\bigg)^{n - 2\ell}\bigg(1 + \frac{z^2}{4}\bigg)^{\ell} + \dfrac{z}{2}\sum_{\ell = 0}^{\lfloor \frac{n-1}{2}\rfloor}
    {n \choose 2\ell + 1}
    \bigg(\frac{z}{2} + 1\bigg)^{n - 2\ell - 1}\bigg(1 + \frac{z^2}{4}\bigg)^{\ell} \>.
\end{align*}
At this point the difficulty with $\Delta = 0$ has been resolved by continuity. We see that $P_n(z)$ is a polynomial of degree $n$.
\end{proof}

% ============================================================================ %
% A Connection with Compositions                                               %
% ============================================================================ %
\section{A Connection with Compositions}
Consider the case with symbolic entries $t_i$, and subdiagonals $-1$ for convenience with minus signs in the formulae.  For instance, the $5$ by $5$ example upper Hessenberg Toeplitz matrix is
\begin{equation}
    \textbf{M}_5 =  \left[ \begin {array}{ccccc} t_{{1}}&t_{{2}}&t_{{3}}&t_{{4}}&t_{{5}}
\\ \noalign{\medskip}-1&t_{{1}}&t_{{2}}&t_{{3}}&t_{{4}}
\\ \noalign{\medskip}0&-1&t_{{1}}&t_{{2}}&t_{{3}}\\ \noalign{\medskip}0
&0&-1&t_{{1}}&t_{{2}}\\ \noalign{\medskip}0&0&0&-1&t_{{1}}\end {array}
 \right] 
\>.
\end{equation}
In this section we consider what happens when we take determinants $P_n(z) = \det(z\mathbf{I} - \textbf{M}_n)$.
Examining $P_0(0)$, $P_1(0)$, $P_2(0)$, $P_3(0)$, and $P_4(0)$, and in particular $P_k(0)$ (i.e. $\det (-\mathbf{M}_k)$) we see that
\begin{align}
    P_0(0) &= 1 \mathrm{\ by\  convention}\\
    P_1(0) &= t_1\\
    P_2(0) &= t_1^2 + t_2\\
    P_3(0) &= t_1^3 + 2t_1t_2 + t_3\\
    P_4(0) &= t_1^4 + 3t_1^2t_2 + 2t_1t_3 + t_2^2 + t_4\>.
\end{align}
One may interpret these (looking at the subscripts) as \textsl{compositions}: $2 = 1 + 1 = 2$; $3 = 1 + 1 + 1 = 1 + 2 = 2 + 1 = 3$; $4 = 1 + 1 + 1 + 1 = 2 + 1 + 1 = 1 + 2 + 1 = 1 + 1 + 2 = 1 + 3 = 3 + 1 = 2 + 2 = 4$. The number of compositions of $n$ is $2^{n-1}$, which we get if all $t_j = 1$.

From the Wikipedia entry on composition (combinatorics), ``a composition of an integer $n$ is a way of writing $n$ as the sum of a sequence of strictly positive integers."

One may interpret the recurrence relation
\begin{equation}
    p_{n, 0} = \sum_{k = 1}^{n}t_kp_{n-k, 0}
\end{equation}
from Proposition~\ref{prop:charpolyrec2_UHT} as saying that to generate a composition of $n$, you get the composition of $n-k$ and then add the number ``$k$" to them; adding these together gives all compositions. For example, when $n = 5$ we have $p_{0, 0} = 1$, $p_{1, 0} = t_1$, $p_{2, 0} = t_1^2 + t_2$, $p_{3, 0} = t_1^3 + 2t_1t_2 + t_3$, and $p_{4, 0} = t_1^{4} + 3t_1^2t_2 + 2t_1t_3 + t_2^2 + t_4$. Then
\begin{align*}
    p_{5, 0} &= t_1 p_{4, 0} + t_2p_{3, 0} + t_3 p_{2, 0} + t_4p_{1, 0} + t_5p_{0, 0} \nonumber\\
    &= t_1^5 + 3t_1^3t_2 + 2t_1^2t_3 + t_1t_2^2 + t_1t_4 + t_2t_1^3 + 2t_1t_2^2 + t_2t_3 + t_1^2t_3 + t_2t_3 + t_4t_1 + t_5 \nonumber \\
    &= t_1^5 + 4t_1^3t_2 + 3t_1^2t_3 + 3t_1t_2^2 + 2t_1t_4 + 2t_2t_3 + t_5 \>.
\end{align*}

\begin{remark}
This determinant also contains the whole characteristic polynomial. Simply replace $t$, with $t_1 - z$ and we get $\det \left(\mathbf{M}_n - z\mathbf{I}\right) = (-1)^{n}P_n$. This suggests that ``compositions with all parts bigger than 1" can be used to generate all compositions. This fact is well-known.
{The combinatorial analysis of this recurrence formula is not quite trivial.}
\end{remark}

\section{Concluding Remarks}
The class of upper Hessenberg Bohemian matrices, and the much smaller class of Bohemian upper Hessenberg Toeplitz matrices, give a useful way to study Bohemian matrices in general. This is an instance of Polya's adage ``find a useful specialization."~\cite[p.~190]{polya2014solve} Because these classes are simpler than the general case, we were able to establish several theorems.

In this paper we have introduced two new formulae for computing the characteristic polynomials of upper Hessenberg Toeplitz matrices. Our first formula, Proposition~\ref{prop:charpolyrec1_UHT}, computes the characteristic polynomials recursively. Our second formula, Proposition~\ref{prop:charpolyrec2_UHT}, computes the coefficients recursively. Finally, we show the number of upper Hessenberg Toeplitz matrices of maximal characteristic height which is at least $2^n$ and we conjecture $\mathcal{O}((1 + \varphi)^{n})$ in Theorem~\ref{thm:count_max_height_uht}.

Many puzzles remain. Perhaps the most striking is the angular appearance of the set of eigenvalues $\mathbf{\Lambda}(\mathbf{M}_n)$, such as in Figures~\ref{fig:UHT_14}, and \ref{fig:UHT_14_0_Diag}. General matrices have eigenvalues asymptotic to a (scaled) disc~\cite{tao2017random}; our computations suggest that as $n \to \infty$, $\sfrac{\mathbf{\Lambda}(\mathbf{M}_{n})}{n^{\sfrac{1}{2}}}$ tends to an irregular hexagonal shape, rather than a disk. More, the density does not appear to be approaching uniformity. Further, the boundary is irregular, with shapes suggestive of what is popularly known as the ``dragon curve" (in reverse---these delineate where the eigenvalues are absent, near the edge). We have no explanation for this.

\section*{Acknowledgements}
The calculations and images presented here were in part made possible using AMD Threadripper workstations provided by the Department of Applied Mathematics at Western University.
We acknowledge the support of the Ontario Graduate Institution, The National Science \& Engineering Research Council of Canada, the University of
Alcal\'a, the Rotman Institute of Philosophy, the Ontario Research Centre of
Computer Algebra, and Western University. Part of this work was developed
while R.~M.~Corless was visiting the University of Alcal\'a, in the frame of the
project Giner de los Rios. L.~Gonzalez-Vega, J.~R.~Sendra and J.~Sendra are
partially supported by the Spanish Ministerio de Econom\'\i a y Competitividad
under the Project MTM2017-88796-P.

\bibliographystyle{siamplain}  % siam
\bibliography{bibliography}

\end{document}